\documentclass[epsf,11pt,times,psfrag,amsmath,amssymb]{article} 

\usepackage{fullpage}
\usepackage{tikz}
\usetikzlibrary{arrows,shapes,snakes,automata,backgrounds,petri}
\usepackage{graphicx,color,float}
\usepackage{epsfig} 
\usepackage{amsmath,amsthm}
\usepackage{amssymb}

\newtheorem{theorem}{Theorem}[section]

\newtheorem{lemma}[theorem]{Lemma}

\newtheorem{definition}{Definition}[section]
\newtheorem{corollary}[theorem]{Corollary}

\usepackage{url}

\title{On Dynamic Optimality for Binary Search Trees\
}
\author{Navin Goyal \footnote{Microsoft Research, India}\\
   {\small\texttt{navingo@microsoft.com}} 
 \and
Manoj Gupta \footnote{IIT Delhi: Work done while visiting Microsoft Research, India}\\
{\small\texttt{gmanoj@cse.iitd.ernet.in}}
}

\date{}                                           

\begin{document}

\maketitle
\begin{abstract}
Does there exist $O(1)$-competitive (self-adjusting) binary search tree (BST) algorithms? This is a well-studied problem. A simple offline 
BST algorithm GreedyFuture was proposed independently by Lucas~\cite{lucas} and Munro~\cite{Munro00}, and they 
conjectured it to be $O(1)$-competitive. Recently, 
Demaine~et~al.~\cite{geom} gave a geometric view of the BST problem. This view allowed them to give an online 
algorithm
GreedyArb with the same cost as GreedyFuture.  However, no $o(n)$-competitive ratio was known for GreedyArb. 
In this paper we make progress towards proving $O(1)$-competitive ratio for GreedyArb by showing that it is 
$O(\log n)$-competitive.  
\end{abstract}

\section{Introduction}
Binary search trees (BST) are a data structure for the dictionary problem. 
There are many examples of (static or offline) binary search trees, e.g. AVL trees and red-black trees, which take $O(\log n)$ worst-case time per search
query (and possibly for other types of operations such as insert/delete; but we will confine ourselves to search queries in this note) on $n$ keys. For static trees, $O(\log n)$ bound cannot be improved. But trees that can change shape in response to queries can potentially have smaller
amortized search time per query. In this case, competitive analysis is used to measure the performance. 

Splay trees~\cite{tarjan} of Sleator and Tarjan are simple self-adjusting binary search trees which are $O(\log n)$-competitive.  Sleator and Tarjan~\cite{tarjan} conjectured that Splay trees were in fact $O(1)$-competitive. Unfortunately, despite considerable efforts $o(\log{n})$-competitiveness for Splay trees is not known; see, e.g., \cite{Pettie08} for a recent discussion. But even a potentially easier question remains open: Is there \emph{any} BST algorithm 
which is $O(1)$-competitive?
In the past decade progress was made on this question and BST algorithms with better competitive ratio were discovered:  
Tango trees~\cite{tango} were the first $O(\log \log n)$-competitive BSTs; Multi-Splay trees~\cite{multisplay} and Zipper trees~\cite{zipper} also have the same competitive ratio along with some additional properties. 
Analyses of these trees use lower bound for the time taken to complete a sequence of requests. Wilber~\cite{wilber} gave two different such lower bounds. Wilber's first lower bound is used in \cite{tango}, \cite{multisplay} and \cite{zipper} to obtain $O(\log \log n)$ competitiveness. These techniques based on Wilber's bound have so far failed to give $o(\log \log n)$-competitiveness.

Even the off-line problem is not well-understood, and the best performance guarantee known is the same as the online guarantee obtained
by Tango trees just mentioned. Lucas~\cite{lucas} and Munro~\cite{Munro00}, independently proposed an offline BST algorithm
GreedyFuture (this name comes from \cite{geom}) and they conjectured that GreedyFuture was $O(1)$-competitive. But even $o(n)$-competitiveness
was not known for GreedyFuture. Our main result, Theorem~\ref{thm} below, implies that GreedyFuture is $O(\log n)$-competitive, thus making progress towards $O(1)$-competitive ratio for GreedyFuture.

A new line of attack on the BST problem was given by Demaine et al.~\cite{geom}. 
They gave a geometric view of the problem of designing BST algorithms. This allowed them 
to translate conjectures and results about BSTs into intuitively appealing geometric statements. We now quickly 
describe their geometric view.

\subsection{Problem Definition} \label{sec:probdef}
In the BST problem we want to maintain $n$ keys in a binary search tree to serve search requests. In response to 
each search request, the algorithm is allowed to modify the structure of the tree, but this change in the structure
adds to the search cost. We are interested in designing 
BST algorithms with small amortized search cost, in other words, algorithms with small competitive ratio. 
Since we will only work in the geometric view, we define the BST problem formally only in the geometric model and 
omit the formal definition of the BST model; please see \cite{geom} for details.

In the geometric view of the BST problem, we will work in the two-dimensional plane with a fixed cartesian coordinate system. We have $n$ keys 
in the tree, which we will assume to be $1, 2, \ldots, n$. And we have $n$ search queries coming at time instants 
$1, 2, \ldots, n$, one for each key (this, of course, is not the general situation as a key could
be searched more than once, but as discussed in \cite{geom}, we can work with this case without loss of generality).
We can represent these $n$ queries as points in the plane:
Let $X=\{p_1, p_2,\dots, p_n\}$ be a set of $n$ points in the two-dimensional plane define as follows. Let the $x$-axis represent the key space and $y$-axis represent time. Each $p_i$ is represented by a pair $(i,t_i)$ where both $i$ and $t_i$ are integers. We say that the key $i$ arrive at time $t_i$. For a point $p$, let $p.x$ denote its $x$-coordinate and $p.y$ denote its $y$-coordinate.  Clearly for 
all distinct $p,q \in X$ we have $p.x \ne q.x$, $p.y \ne q.y$. That is to say, there exists exactly one point from $X$ on line $x=i$, for $1 \le i \le n$. Similarly, there exists exactly one point from $X$ on line $y=i$, for $1 \le i \le n$. For a pair of points $p, q$ not on the same horizontal or vertical line, the axis-aligned rectangle formed by $p$ and $q$ is denoted by $\Box pq$

\begin{definition}[\cite{geom}]
A pair of points $(p,q)$ is said to be \emph{arborally satisfied} with respect to a point set $P$, if (1) $p$ and $q$ lie on the same horizontal and vertical line, or (2) $\exists r \in P \setminus \{p,q\}$ such that $r$ lies inside or on the boundary of
 $\Box pq$. A point set $P$ is \emph{arborally satisfied} if all pairs of points in $P$ are arborally satisfied with 
respect to $P$.
\end{definition}

{\bf Arborally Satisfied Set (ArbSS) Problem}: \emph{Given a point set $X$, find a minimum cardinality  point set $Y$ such that $X \cup Y$ is arborally satisfied.}
[While several definitions here are from \cite{geom}, we have chosen to use less colorful abbreviations than in \cite{geom}.] 
Let MinArb$(X)$ denote the minimum cardinality point set which solves the ArbSS problem on a point set $X$.
\cite{geom} shows that the BST view and the geometric view are essentially equivalent. In particular, if OPT$(S)$ is the minimum cost of 
computing a request sequence $S$, then MinArb$(X)$ = $\Theta($OPT$(S))$, where $X$ is the set of points in the plane corresponding to $S$. 

\subsection{GreedyArb Algorithm}
There is a natural greedy algorithm for the ArbSS problem: 
\begin{quote}
Sweep the point set $X$ with a horizontal line by increasing the $y$-coordinates. Let the point $p$ be processed at time $p.y$. At time $p.y$, place the minimal number of points on line $y=p.y$ to satisfy the rectangles with $p$ as one endpoint and other endpoint in $X \cup Y$ with $y$-coordinate less than $p.y$. This minimal set of point $M_{p}$ is uniquely defined: for any unsatisfied rectangle formed with $(s, p.y)$ as one of the corner point, add a point at $(s, p.y)$. 
\end{quote}

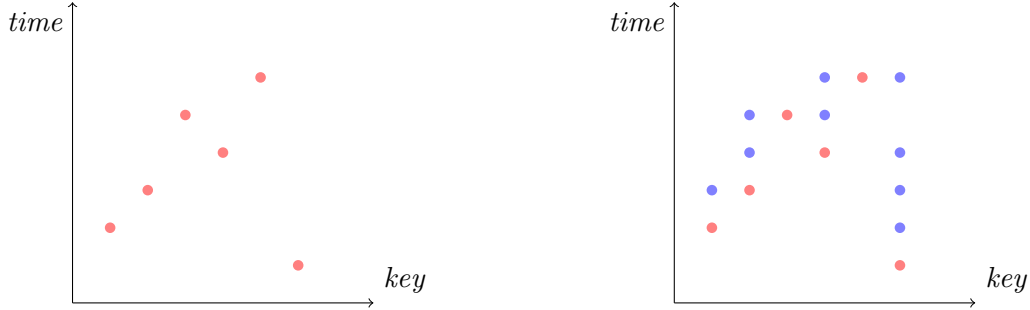
\begin{figure}
\label{figure1}
\begin{tikzpicture}
\draw[->] (0,0) -- (4,0) node[above right] {{\em key}};
\draw[->] (0,0) -- (0,4) node[below left] {{\em time}} ;

\coordinate (A) at (.5,1);
\fill [red,opacity=.5] (A) circle (2pt); 

\coordinate (A) at (1,1.5);
\fill [red,opacity=.5] (A) circle (2pt); 

\coordinate (A) at (1.5,2.5);
\fill [red,opacity=.5] (A) circle (2pt); 

\coordinate (A) at (2,2);
\fill [red,opacity=.5] (A) circle (2pt); 

\coordinate (A) at (2.5,3);
\fill [red,opacity=.5] (A) circle (2pt); 

\coordinate (A) at (3,.5);
\fill [red,opacity=.5] (A) circle (2pt);

\begin{scope}[xshift=8cm]
\draw[->] (0,0) -- (4,0) node[above right] {{\em key}};
\draw[->] (0,0) -- (0,4) node[below left] {{\em time}} ;

\coordinate (A) at (.5,1);
\fill [red,opacity=.5] (A) circle (2pt); 

\coordinate (A) at (1,1.5);
\fill [red,opacity=.5] (A) circle (2pt); 

\coordinate (A) at (1.5,2.5);
\fill [red,opacity=.5] (A) circle (2pt); 

\coordinate (A) at (2,2);
\fill [red,opacity=.5] (A) circle (2pt); 

\coordinate (A) at (2.5,3);
\fill [red,opacity=.5] (A) circle (2pt); 

\coordinate (A) at (3,.5);
\fill [red,opacity=.5] (A) circle (2pt);

\coordinate (A) at (3,1);
\fill [blue,opacity=.5] (A) circle (2pt); 

\coordinate (A) at (.5,1.5);
\fill [blue,opacity=.5] (A) circle (2pt); 

\coordinate (A) at (3,1.5);
\fill [blue,opacity=.5] (A) circle (2pt); 

\coordinate (A) at (1,2);
\fill [blue,opacity=.5] (A) circle (2pt); 

\coordinate (A) at (3,2);
\fill [blue,opacity=.5] (A) circle (2pt); 

\coordinate (A) at (1,2.5);
\fill [blue,opacity=.5] (A) circle (2pt); 

\coordinate (A) at (2,2.5);
\fill [blue,opacity=.5] (A) circle (2pt); 

\coordinate (A) at (3,3);
\fill [blue,opacity=.5] (A) circle (2pt); 

\coordinate (A) at (2,3);
\fill [blue,opacity=.5] (A) circle (2pt);

\end{scope}
\end{tikzpicture}
\caption{The red point are the point set $X$ and the blue points are added by GreedyArb Algorithm}
\end{figure}

GreedyArb is an online algorithm in the sense that at each time instant it adds a set of points so that the resulting set is arborally satisfied,
and it does so without the knowledge of the future requests. \cite{geom} shows that GreedyArb can be used to derive an online BST algorithm with
the same competitive ratio. 
 It is conjectured that GreedyArb is $O(1)$-competitive, that is, the number of points added by the GreedyArb algorithm on a point set $X$ is $O(|MinArb(X)|)$.  Surprisingly, GreedyArb has the same competitive ratio as the aforementioned (offline) BST 
algorithm GreedyFuture \cite{geom}. 
To our knowledge, no non-trivial (i.e. $o(n)$) competitive factor was known for GreedyFuture or GreedyArb. 
In this paper we prove the following result:
\begin{theorem} \label{thm}
GreedyArb is $O(\log n)$-competitive. 
\end{theorem}
Our analysis of GreedyArb exposes some interesting combinatorial properties of GreedyArb algorithm, which may be useful
in further improving the competitive ratio. 
[After we had written our results, it came to our attention that Patrascu in his talk slides~\cite{talk} has claimed that with Iacono he has proved that GreedyArb is $O(\log n)$-competitive. To our knowledge, this result has not appeared anywhere, and our work was done independently.]

\section{Proof of the main result}
In this section we prove Theorem~\ref{thm}. Let us first outline our approach.

Let $X$ be the input set of $n$ points that we want to arborally satisfy by adding more points.  
We want to prove that GreedyArb adds $O(n \log n)$ points. 
We prove this by using the standard recurrence $T(n) = 2T(n/2) + O(n)$, where $T(n)$ is the maximum possible number of 
points added by GreedyArb on sets of $n$ points.  
We interpret the previous equation as follows:  Divide the $n$ points into two equal sets $P$ (points with $x$-coordinate
in $\{1, 2, \ldots, n/2\}$) and $Q$ (rest of the points). 
For sets $P$ and $Q$ we define regions of the plane $R_P$ and $R_Q$ in the natural way: $R_p = \{r \;|\; 1/2 < r.x < n/2+1/2\}$, and similarly $R_Q = \{r \;|\; n/2+1/2 < r.x < n+1/2\}$.
We show that the total number of points added by GreedyArb in $R_p$ when processing points in $Q$ is $O(n)$, and by symmetry, the number of points added in 
$R_q$ when processing points in $P$ is $O(n)$.
This gives the recurrence above for the top level. But we have to prove that the property holds more generally for our recursion to hold at all levels. In general, we show the following: Given any set of $2k$ consecutive keys, with $P$ consisting of the points corresponding to the first $k$ keys and $Q$ consisting of the points corresponding to the last $k$ keys. Then the total number of points added in $R_Q$ by GreedyArb when processing points in $P$ is $O(k)$, and similarly, the 
number of points added in $R_P$ when processing points in $Q$ is $O(k)$. The rest of the proof is devoted to showing this
last statement. Once we prove this, we get the above recurrence and that immediately completes the proof of 
Theorem~\ref{thm}. 
\medskip

We now proceed with the formal proof.

Let $S=\{p_j, p_{j+1}, \dots, p_{j+2k-1}\}$ be a set of $2k$ consecutive points in $X$ such that $p_{i+1}.x = p_i.x +1$ $ \forall j \le i < j+2k-1$. Let $P=\{p_j, p_{j+1}, \dots, p_{j+k-1}\}$ and $Q=\{p_{j+k}, p_{j+k+1}, \dots, p_{j+2k-1}\}$. Define $R_P$ to be the region between the vertical lines passing through points $j-1/2$ and $(j+k-1) + 1/2$; similarly, $R_Q$ is the region between the vertical lines passing through the points $(j+k-1)-1/2$ and $(j+2k-1)+1/2$. 
Let $P_l=\{p_1, p_2, \dots, p_{j-1}\}$ and $Q_r=\{p_{j+2k}, p_{j+2k+1}, \dots, p_{n}\}$. Let $R_{P_l}$ be the region to the left of $R_P$ such that it contains all the point $P_l$. Similarly, $R_{Q_r}$ is the region to the right of $R_Q$ and it contains all the points in $Q_r$. Let $X_{<i}$ denote the set of all the point in $X$ which arrive before time $i$. Let $M_{p}$ be the points added by GreedyArb while processing point $p$. Let $M^P_{p} = \{ m \in M_{p}\;|$ $m$ lies in region $P$\}. \\

\begin{definition}
Let $Z^P_{<t} = \{p \in X_{<t} \;|\; p$ lies in region $R_P\} \cup \{ m \in M_{p} \;|\; p \in X$ and $p.y < t$ and point $m$ lies in the region $R_P\}$. A point $q \in Z^P_{<t}$ is said to be a corner point  in P for $Q$ at time $t$ if there is no point $q' \in Z^P_{<t} \setminus \{q\}$ such that $q'.x \ge q.x$ and $q'.y \ge q.y$.  A point $q \in Z^P_{<t}$ is said to be a corner point  in P for $P_l$ at time $t$ if there is no point $q' \in Z^P_{<t} \setminus \{q\}$ such that $q'.x \le q.x$ and $q'.y \ge q.y$. Let $C_t$ be the set of corner points in $P$ for $Q$ at time $t$.
\end{definition}

\begin{figure}
\label{figure2}
\begin{tikzpicture}
\coordinate (A) at (1.5,1.5);
\fill [black,opacity=0] (A) circle (2pt);
\begin{scope}[xshift=7cm]
\coordinate (A) at (2,.5);
\fill [black,opacity=.5] (A) circle (2pt); 

\coordinate (A) at (1.5,1);
\fill [black,opacity=.5] (A) circle (2pt); 

\coordinate (A) at (0.5,1.5);
\fill [black,opacity=.5] (A) circle (2pt);

\coordinate (A) at (0,2);
\fill [black,opacity=.5] (A) circle (2pt);

\coordinate (A) at (2,1);
\fill [black,opacity=.5] (A) circle (2pt);
\draw [above](2,1) node {\scriptsize{$p_1$}};

\coordinate (A) at (1.5,1.5);
\fill [black,opacity=.5] (A) circle (2pt);
\draw [above](1.5,1.5) node {\scriptsize{$p_2$}};

\coordinate (A) at (.5,2);
\fill [black,opacity=.5] (A) circle (2pt);
\draw [above](.5,2) node {\scriptsize{$p_3$}};

\draw [above](1,3) node {\small{$R_P$}};
\draw [above](4,3) node {\small{$R_Q$}};

\draw[-] (-0.5,0) -- (-0.5,3);
\draw[-] (2.5,0) -- (2.5,3);
\draw[-] (5.5,0) -- (5.5,3);

\draw[-,red] (-1,2.7) -- (6,2.7) ;
\node at(6,2.7)[above right] {$t$};
\end{scope}
\end{tikzpicture}
\caption{$p_1$, $p_2$ and $p_3$ are corner points in $P$ for $Q$ at time $t$}
\end{figure}
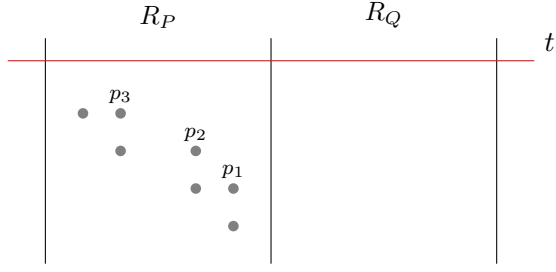

\begin{lemma}
\label{lemma51}
Let $p \in X$ be the point arriving at time $t$. If $p \in P$, then $|C_{t+1}| \le |C_t| + 1$.
\end{lemma}
\begin{proof}
$M^P_{p}$ is the set of points added by GreedyArb in region $R_P$ at time $t$.  Let $r \in (M^P_{p}\cup \{p\})$ be a point such that for all $r' \in (M^P_{p} \cup \{p\})$, $r.x > r'.x$. That is, $r$ is the rightmost point in the set $M^P_{p}\cup \{p\}$. No point in $(M^P_{p} \cup \{p\})\setminus \{r\}$ can be a corner point for $Q$ at time $t$ because each one of them has point $r$ to their right in $R_P$. Only point $r$ can potentially be a corner point among
the points in $M^P_{p}\cup \{p\}$. So the number of corner points can increase by at most one at time $t+1$.
\end{proof}

\begin{lemma}
\label{lemma52}
Let $p \in X$ be the point arriving at time $t$. If $p \in Q \cup Q_r$, then (1) if $|M^P_p| = 0$, then $|C_{t+1}| = |C_t|$; and (2) if $|M^P_p| > 0$, then $|C_{t+1}| \le |C_t| - (|M^P_{p}| -1)$.
\end{lemma}
\begin{proof}
If $|M^P_p| = 0$, then the number of corner points cannot increase as there are no points on line $y=t$ in region $R_P$.  So let $|M^P_p| >0$. The execution of GreedyArb tries to satisfy all the unsatisfied rectangles with one corner $p$. For each marked point in $M^P_p$, the corresponding other corner of the rectangle must lie in $C_t$. Let $N^P_{p} =  \{ q \;|\; q \in C_t$ and GreedyArb adds a point $r \in M^P_{p}$ due to unsatisfied rectangle $\Box qp\}$. At time $t+1$, all the points in $N^P_{p}$ cease to remain the corner point, because $\forall q \in N^P_{p}$ $\exists r \in M^P_{p}$ such that $r.x = q.x$ and $r.y > q.y$. So the decrease in the number of corner points is at least $|N^P_{p}| = |M^P_{p}|$. As in the proof of Lemma~\ref{lemma51}, only one point in $M^P_{p}$ can become a corner point for $Q$ at time $t+1$. So $C_{t+1} \le C_t - ( |M^P_{p}| -1 )$.
\end{proof}

Assume that $P_l = \emptyset$. Let $p$ be the last point to arrive in set $Q$. Applying Lemmas~\ref{lemma51} and \ref{lemma52} inductively we get, $|C_{p.y}| \le |C_1| + \displaystyle\sum_{q \in P} 1- \sum_{\substack{q \in Q \cup Q_r \\M^P_{q} \ne \phi}} (|M^P_{q}| -1)$. Since $|C_1| = 0$ this gives, 

\begin{align*}
\sum_{\substack{q \in Q \cup Q_r \\ M^P_{q} \ne \emptyset}} (|M^P_{q}| -1) \leq -|C_{p.y}| +  |P|,
\end{align*}

this gives 

\begin{align*}
\sum_{\substack{q \in Q \\ M^P_{q} \ne \emptyset}} (|M^P_{q}| -1) \leq -|C_{p.y}| +  |P|,
\end{align*}

and so 
\begin{align*}
\sum_{\substack{q \in Q \\ M^P_{q} \ne \emptyset}} |M^P_{q}| \leq  |P| + \sum_{\substack{q \in Q \\ M^P_{q} \ne \phi}} 1  \leq |P| + |Q| \leq 2k.
\end{align*}


So if $P_l = \emptyset$, then the number of points added while processing points in $Q$ in region $R_P$ is $O(k)$. But if $P_l \ne \emptyset$, then the quantity $\sum_{q \in Q}|M^P_{q}|$ may be larger. We will argue that it's still $O(k)$. We denote the set of all the points added by $GreedyArb$ as $Y$.

\begin{definition}
For a point $p\in P$ let $T^p_{<t}$ be all the points from $X \cup Y$ on line $x=p.x$ with their $y$-coordinate less than $t$, where $t > p.y$. Let $q$ be the point in $T^p_{<t}$ with largest $y$-coordinate. We say that point $p$ is \emph{hidden} at time $t$ if there exist points $q_l, q_r \in R_p$ respectively to the left and right of $q$ on line $y=q.y$. If point $p$ is hidden at time $t$, point $q$ cannot be a corner point for $Q$ at time $t$.
\end{definition}

\begin{definition}
For a point $p \in P$ let $r$ be the point in $T^p_{<t}$ with the largest $y$-coordinate. We say that the point $p$ is 
\emph{exposed} at time $t$ if 
either there is no point to the left or no point to the right of point $r$ on the line $y=r.y$ in region $R_p$ (in other
words, at least one (or both) of the two sides (left and right) is empty). 
If $p$ is exposed at time $t$, then point $r$ may potentially be a corner point for $Q$ at time $t$.
\end{definition}

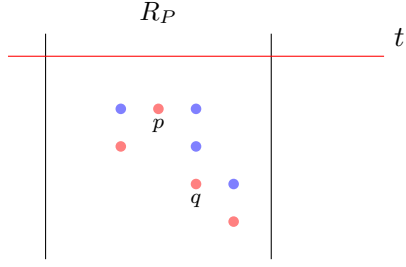
\begin{figure}
\label{figure3}
\begin{tikzpicture}
\coordinate (A) at (1.5,1.5);
\fill [red,opacity=0] (A) circle (2pt);
\begin{scope}[xshift=8cm]
\coordinate (A) at (2,.5);
\fill [red,opacity=.5] (A) circle (2pt); 

\coordinate (A) at (1.5,1);
\fill [red,opacity=.5] (A) circle (2pt); 

\coordinate (A) at (0.5,1.5);
\fill [red,opacity=.5] (A) circle (2pt);

\coordinate (A) at (1,2);
\fill [red,opacity=.5] (A) circle (2pt);

\coordinate (A) at (2,1);
\fill [blue,opacity=.5] (A) circle (2pt);

\coordinate (A) at (1.5,1.5);
\fill [blue,opacity=.5] (A) circle (2pt);

\coordinate (A) at (.5,2);
\fill [blue,opacity=.5] (A) circle (2pt);

\coordinate (A) at (1.5,2);
\fill [blue,opacity=.5] (A) circle (2pt);

\draw [above](1,2) node [below]{\scriptsize{$p$}};
\draw [above](1.5,1) node [below]{\scriptsize{$q$}};

\draw [above](1,3) node {\small{$R_P$}};

\draw[-] (-0.5,0) -- (-0.5,3);
\draw[-] (2.5,0) -- (2.5,3);

\draw[-,red] (-1,2.7) -- (4,2.7) ;
\node at(4,2.7)[above right] {$t$};
\end{scope}
\end{tikzpicture}
\caption{At time $t$ , $p \in X$ is hidden but $q \in X$ is exposed}
\end{figure}

\begin{lemma}
\label{lemma53}
Let $p \in P$ be hidden at time $t$ and let it be exposed for the first time at time $t'+1$. Let $q$ be the point processed by GreedyArb at time $t'$, then $q \in P$.
\end{lemma}
\begin{proof}
Since $p$ remains hidden at time $t'$, $\exists r$ such that $r$ has the largest $y$-coordinate in $T^p_{<t'}$ and there exist points $r_l$ and $r_r$ respectively to the left and right of $r$ on line $y=r.y$ in $R_P$. Also since $p$ is exposed for the first time at time $t'+1$, there must be a point added by GreedyArb at $(p.x,t')$. But if $q \in P_l$, then the rectangle $\Box rq$ is already satisfied by $r_l$. Similarly, if $q \in (Q \cup Q_r)$, then the rectangle $\Box rq$ is satisfied by $r_r$. So if $q \in (P_l \cup Q \cup Q_r)$, then GreedyArb will not add any point at $(p.x,t')$. So $q \in P$.
\end{proof}

\begin{corollary}
\label{coro54}
A point $p\in P$ can be hidden by any other point $q \in X$, but it can be exposed by a point in $P$ only. Also let  $p \in P$ be hidden at time $t$ and be exposed for the first time at time $t'+1$. Let $q \in (P_l \cup Q \cup Q_r)$ with $t < q.y < t'$, then GreedyArb cannot add any point at $(p.x,q.y)$. 
\end{corollary}

\begin{definition}
For each $q \in M^P_{p}$, there exists a unique point $r \in X$ (we get the uniqueness because of our assumption that
on each vertical line $X$ has at most one point) such that $r.y < q.y$ and $r.x=q.x$. Let r be called the $parent$ to $q$ and be denoted as $parent(q)$. 
\end{definition}

Let $q_l$ and $q_r$ be the points in $M^P_{p}$ with the smallest and largest $x$-coordinates respectively. By definition, for each point  $q \in M^P_{p} \setminus \{q_r,q_l\}$, $parent(q)$ either remains hidden or is hidden at time $p.y+1$. Only $parent(q_l)$ and $parent(q_r)$ may become exposed at time $p.y+1$. So a point $p \in P$ can expose at most two points in $P$.

\begin{lemma}
\label{lemma55}
The total number of times the points in $P$ change their state from hidden to exposed and vice versa is at most $5k$.
\end{lemma}
\begin{proof}
The point $p$ may initially be exposed. Once $p$ is hidden, it remains hidden till a point $q \in P$ exposes it. Each point $q \in P$ can expose at most two points of $P$. So the total number of points that can be exposed is at most $2k$. After that, only $k$ points can become hidden again as $|P| = k$, but these points cannot get exposed again. So the number of times points in $P$ change their state = $2k$ (exposed to hidden) $+ 2k$ (hidden to exposed) $+k$ (exposed to hidden) $= 5k$.
\end{proof}

Remark: This is a very crude analysis of the total number of times the points $P$ can change their state. The number of times the state changes is probably at most $2k$.
\begin{lemma} \label{lem:main}
 $\displaystyle\sum_{p \in Q} |M^P_{p}| \le 7k$.
\end{lemma}
\begin{proof}
Assume for contradiction $\displaystyle\sum_{p \in Q} |M^P_{p}| > 7k$. Let the two points in $M^P_{p}$ with the smallest and largest $x$-coordinates be called extreme points. Let $Z$ be the set of all extreme points in the set $\displaystyle\cup_{p \in Q} M^P_{p}$. So $|(\displaystyle\cup_{p \in Q} M^P_{p} )\setminus Z |> 5k$. By Corollary~\ref{coro54}, no point can be added below any point already hidden while processing $p$. So for each point $q \in (\displaystyle\cup_{p \in Q} M^P_{p}) \setminus Z$, $parent(q)$ must be exposed at time $p.y$ and gets hidden at time $p.y+1$. By Lemma~\ref{lemma55}, the number of times the points in $P$ can change state from exposed to hidden is at most $5k$. So this gives a contradiction. So $\displaystyle\sum_{p \in Q} |M^P_{p}| \le 7k$
\end{proof}

Thus we have shown that the number of points in $R_P$ added by GreedyArb when processing points in $Q$ is $O(k)$. It 
follows by symmetry that the number of points in $R_Q$ added by GreedyArb when processing points in $P$ is also $O(k)$. 

Let $T_{X[1\dots n]}$ be the total number of points added by GreedyArb algorithm. More generally, for $1 \leq i < j \leq n$,
let $T_{X[i,j]}$ be the maximum possible number of points added by GreedyArb in region $R_{X[i,j]}$ when processing points in set $X[i,j]$,
where $R_{X[i,j]}$ is the region between the vertical lines passing through $i-1/2$ and $j+1/2$. 
Here the maximum is taken over all
possible sets $X$ of size $n$, satisfying our assumption in Sec.~\ref{sec:probdef} (namely, each vertical line has at most
one point of $X$, and each horizontal line has at most one point of $X$). Now we can write the recurrence introduced at the
beginning of the proof. We divide the $2k$ points into two set. Set $P=X[j \dots k-1]$ contains the first half and the set 
$Q=X[j+k \dots j + 2k-1]$ contains the second half and $R_P$ and $R_Q$ are the corresponding regions.\\
\begin{tabular}{lll}
$T_{X[j\dots j+2k-1]}$ &$=$&  the total number of points added when processing points of $P$ in $R_Q$ $(\sum_{p \in P} |M^Q_{p}| )$ + \\
		&&    the total number of points added when processing points of $Q$ in $R_P$ $(\sum_{p \in Q} |M^P_{p}| )$ + \\
		&&    the total number of points added when processing points of $P$ in $R_P$ $(T_{X[j \dots j+k-1])}$) + \\ 
		&&   the total number of points added when processing points of $Q$ in $R_Q$ $(T_{X[j+k \dots j+2k-1]})$. \\ \\
\end{tabular}
Using Lemma~\ref{lem:main} this gives  $T_{X[j\dots j+2k-1]}$ = $T_{X[j \dots j+k-1])}$ +  $T_{X[j+k \dots j+2k-1]}$ + $O(k)$.\\
Which gives our desired result:  $T_{X[1\dots n]}$ = $ O(n \log n)$.

\bibliographystyle{plain}
\nocite{*}
\bibliography{arboral}
\end{document}